\documentclass[pre,nobalancelastpage,danish,twocolumn,nofootinbib,superscriptaddress]{revtex4-1}
\pdfoutput=1

\usepackage{amsmath}
\usepackage{amssymb}
\usepackage{graphicx}
\usepackage{dcolumn}
\usepackage{bm}

\usepackage[utf8]{inputenc}
\usepackage[T1]{fontenc}
\usepackage{mathptmx}
\usepackage{subeqnarray}
\usepackage[colorlinks=true,linkcolor=blue, citecolor=blue, urlcolor=blue, bookmarks]{hyperref}
\usepackage{amsthm}

\newtheorem{thm}{Theorem}

\theoremstyle{definition}
\newtheorem{defn}{Definition}[section]

\theoremstyle{remark}
\newtheorem{rem}{Remark}

\newcommand{\ii}{\mathrm{i}}
\newcommand{\dd}{\mathrm{d}}
\newcommand{\var}{\mathrm{Var}}

\begin{document}

\title[Microscopic Dynamics of Nonlinear Fokker-Planck Equations]{Microscopic Dynamics of Nonlinear Fokker-Planck Equations}

\author{Leonardo Santos}
\email{leo\_vieira@usp.br}
\affiliation{ 
Departamento de F\'{i}sica-Matem\'{a}tica, Instituto de F\'{i}sica,
Universidade de S\~{a}o Paulo, S\~{a}o Paulo 05508-090, S\~{a}o Paulo, Brazil
}%

\date{\today}

\begin{abstract}
We propose a new approach to describe the effective microscopic dynamics of (power-law) nonlinear Fokker-Planck equations. Our formalism is based on a nonextensive generalization of the Wiener process. This allow us to obtain, in addition to significant physical insights, several analytical results with great simplicity. Indeed, we obtain analytical solutions for a nonextensive version of Brownian free-particle and Ornstein-Uhlenbeck process, and explain anomalous diffusive behaviours in terms of memory effects in a nonextensive generalization of Gaussian white noise. Finally, we apply the develop formalism to model thermal noise in electric circuits.
\end{abstract}

\maketitle

Fokker-Planck equations (FPEs) constitute a powerful tool in the study of nonequilibrium phenomena \cite{Risken96}. Since the seminal contribution by Einstein on the Brownian motion \cite{Einstein}, linear equations of this class play a fundamental role in the study of normal diffusion processes and in the investigation of nonequilibrium properties in general. It is a well-known fact, however, that diverse phenomena in complex systems are associated to an anomalous diffusive behavior that cannot be properly described by linear FPEs. Indeed, for various applications including diffusion in porous media \cite{Aronson}, type-II superconductors \cite{RNC12}, granular media \cite{CRSA15}, and self-gravitating systems \cite{Shiino03,Chavanis03}, a nonlinear FPE \cite{Frank05} appears to be more suitable.

In the present work, we consider nonlinear FPEs of the form
\begin{eqnarray}\label{eq:1}
\frac{\partial p}{\partial t}=-\mathrm{Div}(\mathbf{F}p)+D\Delta p^\nu,
\end{eqnarray}
where $p(\mathbf{x},t)$ is a time-dependent density, $\mathbf{x}\in\mathbb{R}^d$ represents a point in an adequate configuration space, $D>0$ is a diffusion constant, $\mathbf{F}(\mathbf{x},t)$ is a drift force, and $\nu:=2-q$ is a real parameter characterizing the nonlinearity appearing in the diffusion term. The power-law nonlinear FPEs (\ref{eq:1}) have a deep connection with the nonextensive entropy $S_q$ \cite{PP95,TB96,Tsallis88}, and for that they have an almost paradigmatic role in nonextensive statistical mechanics \cite{Tsallisbook}. Moreover, such equations have properties that are relevant and/or interesting from both a physical and mathematical point of view. For instance, they admit exact analytical solutions of the $q$-Gaussian form that can be interpreted as maximum entropy densities obtainable from the optimization (under appropriate constraints) of $S_q$; they have been studied in the context of entropy production \cite{RCN15,casasthese}; with different drift forces \cite{WPT16,PWCNT18}; they obey an $H$-theorem formulated in terms of a free-energy-like quantity \cite{SNC07,SRT16}, and so on. In particular, we highlight an experimental work on granular media \cite{CRSA15}, that verified within great precision ($2\%$ error) the Tsallis and Bukman's \cite{TB96} prediction 19 years after the original proposal.

Nonextensive statistical mechanics is a (possible) generalization of the Boltzmann-Gibbs theory that aims to extend its domain of applicability to phenomena with long-ranged interactions and memories, and (multi)fractal configuration spaces. This theory has been satisfactorily applied to handle a large number of physical phenomena \footnote{A selection of applications and verifications of nonextensive statistical mechanics may be found at:
\href{http://tsallis.cat.cbpf.br/experimental.htm}{http://tsallis.cat.cbpf.br/experimental.htm}}. Its development provides, in addition to several remarkable physical insights, a powerful mathematical formalism that has been extensively studied in recent years.  One of the main goals of such a formalism is to generalize mathematical concepts and tools in order to simplify the formal treatment of problems arising from nonextensive systems. Remarkable examples are the $q$-generalizations of the usual transcendent functions (exponential, sine, cosine, \dots) \cite{Borges98,Yamano02}; the $q$-Fourier transform \cite{UTS08}; a generalization of the central limit theorem \cite{UTS08}; $q$-Dirac's delta functions \cite{JT10,ST17}; and so on.

An interesting problem in the study of nonequilibrium systems is the reconstruction of the microscopic dynamics from FPEs. For the linear case ($\nu=q=1$), this problem is simple and the associated microscopic dynamics is governed by the usual Langevin equation \cite{TO15}. For the nonlinear case, however, this problem is not trivial. In this regard, Borland \cite{Borland98} proposed a phenomenological model in which the microscopic dynamics successfully reproduce (\ref{eq:1}). The equations of motion, however, depend on the solution of the nonlinear FPE itself. In other words, there exists a coupling between the macroscopic description and the microscopic one.

In this paper we propose a new approach to the microscopic dynamics of nonlinear FPEs (\ref{eq:1}). For that, we conjecture the existence of a stochastic process $\mathbf{P}_t^q$ in such a way that the solution of the following stochastic differential equation:
\begin{eqnarray}\label{eq:2}
\dd \mathbf{X}_t=\mathbf{F}(\mathbf{X}_t,t)\dd t+(2D)^\frac{1}{2+d(1-q)} \dd \mathbf{P}_t^{q},
\end{eqnarray}
is a stochastic process whose probability distribution $p$ is a solution of (\ref{eq:1}). In other words, Eq. (\ref{eq:2}) reproduces the microscopic dynamics of nonlinear FPEs (\ref{eq:1}). The process $\mathbf{P}^q_t$ may be understood as a nonextensive generalization of the Wiener process. This new formalism completely removes the dependence between microscopic dynamics and nonlinear FPEs. As will be shown, several physically relevant and mathematically interesting insights follow from this result. Namely, we compute analytical solutions for (\ref{eq:1}) with $\mathbf{F}=\mathbf{0}$ and $\mathbf{F}=-\gamma \mathbf{x}$ in a very simple and economical way; we define a nonextensive analogue of the Gaussian white noise; and finally we apply the developed formalism for modeling thermal noises in electric circuits.

We have organized the paper as follows: in section \ref{sec:II} we review the mathematical framework of nonextensive statistical mechanics; in section \ref{sec:III} we present our results; and we finish the paper with a discussion in section \ref{sec:IV}.

\section{Nonextensive Formalism}\label{sec:II}

Herein we will briefly review the mathematical framework of nonextensive statistical mechanics (for a more detailed discussion see Chapter 3 of Ref. \cite{Tsallisbook}). The main goal is to introduce three concepts: $q$-Gaussian distributions, $q$-Fourier transform, and $q$-independence.

The starting point of our discussion is the \textit{$q$-exponential function},
\begin{eqnarray}\label{eq:3}
e_q^x := \begin{cases}
[1+(1-q)x]^{\frac{1}{1-q}} \hspace{0.5cm} &\textrm{if } 1+(1-q)x\geq 0 \\ 0  &\textrm{otherwise}
\end{cases},
\end{eqnarray}
where $q,x\in \mathbb{R}$. For a pure imaginary variable, $e_q^{\ii x}$ can be defined to be the principal value of $[1+(1-q)\ii x]^{\frac{1}{1-q}}$. Since $e_q^x\to e^x$ when $q\to 1$, the $q$-exponential function may be understood as a generalization of the usual exponential function $e^x$. The inverse of $e_q^x$, the $q$-logarithm $\log_q(x)$, reads
\begin{eqnarray}\label{eq:4}
\log_q(x):=\frac{x^{1-q}-1}{1-q},
\end{eqnarray}
for $x>0$. 

The $q$-exponential function and its inverse have several useful and interesting properties \cite{Yamano02}. Unfortunately, a very useful property of the exponential, $e^{a+b}=e^a e^b$, does not hold for $q$-exponentials with $q\neq 1$. In order to get around this difficulty, we may define the operations $\oplus_q$ and $\otimes_q$,
\begin{eqnarray}\label{eq:5}
a\oplus_q b:= a+b+(1-q)ab
\end{eqnarray}
and 
\begin{eqnarray}\label{eq:6}
a\otimes_q b :=\begin{cases}[a^{1-q}+b^{1-q}-1] \hspace{0.5cm} &\textrm{if } a^{1-q}+b^{1-q}-1\geq 0 \\ 0 &\textrm{otherwise}\end{cases},
\end{eqnarray}
in such a way that
\begin{eqnarray}\label{eq:7}
e_q^a\otimes_q e_q^b=e_q^{a+b} \hspace{1cm} \textrm{and} \hspace{1cm} e_q^{a\oplus_q b}=e_q^a e_q^b.
\end{eqnarray}
As a consequence of the properties above, for the $q$-logarithm we have
\begin{eqnarray}\label{eq:8}
\log_q(a\otimes_q b)=\log_q a + \log_q b.
\end{eqnarray}

Another important concept in our discussion is that of a \textit{$q$-Gaussian distribution}. As the name suggests, these are $q$-generalizations of Gaussian distributions. We say that a random variable $\mathbf{X}$ in $\mathbb{R}^d$ is $q$-Gaussian if its probability distribution is
\begin{eqnarray}\label{eq:9}
G_q(\mathbf{x})=\left(\frac{d}{\mathsf{2}_{qd}\pi_{qd} \sigma^2}\right)^{\frac{d}{2}}\exp_{q}\left(-\frac{d}{\mathsf{2}_{qd}\sigma^2}|\mathbf{x}-\boldsymbol{\mu}|^2\right).
\end{eqnarray}
The constants $\mathsf{2}_{qd}$ and $\pi_{qd}$ are defined as
\begin{eqnarray}\label{eq:10}
\mathsf{2}_{qd}=\begin{cases}2-(d+2)(q-1) \hspace{0.5cm} &\textrm{if } q>1\\ 2(2-q)+d(1-q) &\textrm{otherwise} \end{cases},
\end{eqnarray}
and
\begin{eqnarray}\label{eq:11}
\pi_{qd}=\frac{\pi |1-q|^{-1}}{\left[\Gamma\left(\frac{d}{2}\right)\right]^{\frac{2}{d}}}\begin{cases}\left[\textrm{B}\left(\frac{d}{2};\frac{2+d(1-q)}{2(q-1)}\right)\right]^{\frac{2}{d}} \hspace{0.5cm} &\textrm{if } q>1 \\ \left[\textrm{B}\left(\frac{d}{2};\frac{2-q}{1-q}\right)\right]^{\frac{2}{d}} &\textrm{otherwise}\end{cases},
\end{eqnarray}
where $\mathrm{B}$ and $\Gamma$ denote the Beta and Gamma functions respectively. Defined in this way, it is possible to demonstrate that: (\textit{i}) $G_q$ is normalized for all $q<(2+d)/d$; (\textit{ii})
$\boldsymbol{\mu}$ is the $q$-expectation value of $\mathbf{X}$, \begin{eqnarray}\label{eq:12}
\mathbb{E}_q[\mathbf{X}]:=\left(\int_{\mathbb{R}^d}[p(\mathbf{x})]^q\dd \mathbf{x}\right)^{-1}\int_{\mathbb{R}^d} \mathbf{x}[p(\mathbf{x})]^q\dd \mathbf{x};
\end{eqnarray}
(\textit{iii}) $\sigma^2$ is the variance of $\mathbf{X}$,
\begin{align}\label{eq:13}
\var[\mathbf{X}]:=&\mathbb{E}[\mathbf{X}^2]-(\mathbb{E}[\mathbf{X}])^2 \nonumber \\ =&\int_{\mathbb{R}^d} |\mathbf{x}|^2p(\mathbf{x})\dd \mathbf{x}-\left|\int_{\mathbb{R}^d} \mathbf{x} p(\mathbf{x})\dd \mathbf{x}\right|^2,
\end{align}
which is finite whenever $q<(d+4)/(d+2)$.

The $q$-Gaussian distributions are deeply related with the nonextensive entropy $S_q$. We define the $q$-entropy of a random variable $\mathbf{X}$ taking values $\mathbf{x}\in\mathbb{R}^d$ with probability $p(\mathbf{x})$ as the nonextensive generalization of Shannon's entropy, \textit{i.e.}
\begin{eqnarray}\label{eq:14}
S_q[\mathbf{X}]:=\int_{\mathrm{Supp}p} p(\mathbf{x})\log_q\left(\frac{1}{p(\mathbf{x})}\right) \dd \mathbf{x}.
\end{eqnarray}
It is possible to demonstrate that $S_q$ is maximized (under appropriated constrains) for a $q$-Gaussian random variable. In addition, $S_q$ is not additive but $q$-additive. 

An useful property of the Fourier transform is that it maps a Gaussian distribution with variance $\sigma^2$ into another Gaussian distribution with variance $\propto 1/\sigma^2$ \cite{TO15}. In general, the Fourier transform of a $q$-Gaussian does not correspond to another $q$-Gaussian. We define the \textit{$q$-Fourier transform} as an integral transform that maps $q$-Gaussians into $\tilde{q}$-Gaussians (in general we do not require $\tilde{q}=q$) \cite{UTS08}. The $q$-Fourier transform of a non-negative measurable function $f$, denoted by $\hat{f}_q$, is defined, for $1\leq q<(2+d)/d$, as
\begin{eqnarray}\label{eq:15}
\hat{f}_q(\mathbf{k})=\int_{\mathrm{Supp} f} f(\mathbf{x})\otimes_q e_q^{\ii\mathbf{x}\cdot \mathbf{k}}\dd \mathbf{x}.
\end{eqnarray}
If $\mathbf{X}$ is a random variable, we define the \textit{$q$-characteristic function of $\mathbf{X}$}, denoted by $F_q[\mathbf{X}]$, as the $q$-Fourier transform of its probability distribution.  

The $q$-Fourier transform is the nonextensive generalization of the usual Fourier transform. However, this generalization does not have all the good properties of its extensive counterpart. The main problem refers to its invertibility. Indeed, it is not invertible in the full space of probability density functions for $q>1$ \cite{Hilhorst10}. In the space of $q$-Gaussian distributions, however, the $q$-Fourier transform defines an injective map. In fact, as a straightforward generalization of the Lemma 2.5 of Ref. \cite{UTS08}, one may verify that if $q\in[1,(d+4)/(d+2))$, then the $q$-Fourier transform maps a $q$-Gaussian into a $\tilde{q}$-Gaussian (up to normalization), with $\tilde{q}=[2+(d-2)(1-q)]/[2+d(1-q)]$. More precisely, if $\mathbf{X}\sim \mathcal{N}_q(\mathbf{0},\sigma^2)$ in $\mathbb{R}^d$, then
\begin{align}\label{eq:16}
F_q[\mathbf{X}]=\exp_{\tilde{q}}\left[-\frac{\mathsf{2}_{qd}[2+d(1-q)]}{8(\mathsf{2}_{qd}\pi_{qd})^{d(q-1)}}\left(\frac{d}{\sigma^2}\right)^{d(q-1)-1}|\mathbf{k}|^2\right].
\end{align}

Finally, the \textit{$q$-independence} consists of a mathematical property associated with two or more random variables. We say that two random variables, $\mathbf{X}$ and $\mathbf{Y}$, are $q$-independent if the $q$-characteristic function of $\mathbf{Z}:=\mathbf{X}+\mathbf{Y}$ can be written as
\begin{eqnarray}\label{eq:17}
F_q[\mathbf{Z}](\mathbf{k})=F_{\tilde{q}}[\mathbf{X}](\mathbf{k}) \otimes_{\tilde{q}} F_{\tilde{q}}[\mathbf{Y}](\mathbf{k}),
\end{eqnarray}
where $\tilde{q}=[2+(d-2)(1-q)]/[2+d(1-q)]$. It is easy to verify that this notion of independence reduces to the usual one in the extensive limit. In fact, from (\ref{eq:17}), when $q\to 1$ we have that the characteristic function of $\mathbf{Z}$ is the product between the characteristic function of $\mathbf{X}$ and $\mathbf{Y}$, meaning that they are statistically independent \cite{TO15}.

\section{Microscopic Dynamics of Nonlinear FPEs}\label{sec:III}

In this section we will discuss a possible way to describe the microscopic dynamics of nonlinear FPEs. As already mentioned, for linear FPEs ($\nu=q=1$) the microscopic dynamics is described by the usual Langevin equation. Using stochastic calculus notation, this equation reads
\begin{eqnarray}\label{eq:18}
\dd \mathbf{X}_t=\mathbf{F}(\mathbf{X}_t,t)\dd t+\sqrt{2D}\dd \mathbf{W}_t,
\end{eqnarray}
where $\mathbf{W}_t$ denotes the Wiener process in $\mathbb{R}^d$. The connection between (\ref{eq:18}) and linear FPEs may be found using It\^{o}'s lemma (see, for instance, Refs. \cite{Gardiner,casasthese} for detailed discussions).

It is a remarkable feature of the microscopic dynamics of linear FPE has such a simple form. In particular, we emphasize that the equation of motion (\ref{eq:18}) has no explicit dependence on the density $p$, and therefore can be solved without any reference to the corresponding FPE. For the nonlinear case, however, if one tries to describe the microscopic dynamics in terms of $\mathbf{W}_t$, the following equation is obtained \cite{Borland98}
\begin{eqnarray}\label{eq:19}
\dd \mathbf{X}_t=\mathbf{F}(\mathbf{X}_t,t)\dd t+\sqrt{2D}[p(\mathbf{X}_t,t)]^{\frac{\nu-1}{2}} \dd \mathbf{W}_t.
\end{eqnarray}
This shows that the Wiener process is not adequate to describe the microscopic dynamics of the nonlinear FPEs. We may attribute this inadequacy to the fact that $\mathbf{W}_t$ is Gaussian and has independent increments, \textit{i.e.} for every $t>0$ the future increments $\mathbf{W}_{t+h}-\mathbf{W}_{t}$, $h\geq 0$, are statistically independent of the past values $\mathbf{W}_s$, $s\leq t$ \cite{Gardiner}. It is currently known, however, that the solutions of (\ref{eq:1}) are $q$-Gaussians \cite{TB96}, and their connection with nonextensive entropy $S_q$ makes us intuit that the assumption of independent increments should not be adequate.

Based on nonextensive formalism in the previous section, we define a $q$-generalization of the Wiener process as follows.

\begin{defn}\label{defn:1} The $d$-dimensional nonextensive Wiener process $\mathbf{P}_t^q$ is a stochastic process in $\mathbb{R}^d$ defined by the following properties:
\begin{enumerate}
\item $\mathbf{P}_0^q=\mathbf{0}$ almost surely;
\item The paths $t\mapsto\mathbf{P}_{t}^q$ are continuous with probability $1$;
\item $\mathbf{P}_t^q$ has $q$-independent increments, \textit{i.e.} $\mathbf{P}_{t}^q-\mathbf{P}_{t'}^q$ is statistically independent of $\{\mathbf{P}^q_s:s\leq t'\}$ for any $0\leq t'\leq t$; 
\item Given $t$ and $t'$, $0\leq t'<t$, $\mathbf{P}_{t}^q-\mathbf{P}_{t'}^q$ is a $q$-Gaussian random variable, with
\begin{eqnarray}\label{eq:20}
\mathbb{E}_q\left[\mathbf{P}_{t}^q-\mathbf{P}_{t'}^q\right]=\mathbf{0},
\end{eqnarray}
and
\begin{eqnarray}\label{eq:21}
\var\left[\mathbf{P}_{t}^q-\mathbf{P}_{t'}^q\right]=\varepsilon_{qd} (t-t')^{\frac{2}{2+d(1-q)}},
\end{eqnarray}
where
\begin{eqnarray}\label{eq:22}
\varepsilon_{qd}=\frac{d}{\mathsf{2}_{qd}\pi_{qd}}\left[\pi_{qd}(2-q)(2+d(1-q))\right]^{\frac{2}{2+d(1-q)}}.
\end{eqnarray}
\end{enumerate}
\end{defn}

In what follows, we will consider that $\mathbf{P}_t^q$ is a well-defined stochastic process for all $q<(d+4)/(d+2)$. It should be stressed that $\mathbf{P}_t^q$ is the usual Wiener process in the limit $q\to 1$. The main differences between $\mathbf{W}_t$ and $\mathbf{P}_t^q$ are the following: (\textit{i}) $\mathbf{W}_t$ is Gaussian, while $\mathbf{P}_t^q$ is $q$-Gaussian; (\textit{ii}) the increments of $\mathbf{W}_t$ are independent, while the increments of $\mathbf{P}_t^q$ are $q$-independent; (\textit{iii}) the variance of $\mathbf{W}_t$ is $\var[\mathbf{W}_t]=t$, while $\var[\mathbf{P}_t^q]\propto t^\alpha$, $\alpha=2/[2+d(1-q)]$.

Once the nonextensive analogue of the Wiener process is well-defined, we consider the following stochastic differential equation:
\begin{eqnarray}\label{eq:23}
\dd \mathbf{X}_t=\mathbf{F}(\mathbf{X}_t,t) \dd t+(2D)^{\frac{1}{2+d(1-q)}} \dd \mathbf{P}_t^q,
\end{eqnarray}
associated to the nonlinear FPEs (\ref{eq:1}). Since $\mathbf{P}_t^q\to \mathbf{W}_t$ when $q\to 1$, Eq. (\ref{eq:23}) is a generalization of Eq. (\ref{eq:18}). We conjecture that the microscopic dynamics of the nonlinear FPE (\ref{eq:1}) is governed by Eq. (\ref{eq:23}). In order to demonstrate that such a conjecture is reasonable, we will consider some particular cases whose results are known.

We define a nonextensive analogue of the stochastic integral inspired by the definition of It\^{o}'s integral (see, for instance, Ref. \cite{Gardiner}) as follows.

\begin{defn}\label{defn:2}
Let $\Pi=\{t_0,t_0+\Delta,\dots,t_0+n\Delta=t\}$ ($\Delta=(t-t_0)/n$) be a homogeneous partition of the interval $[t_0,t]$. The $q$-Itô's integral of a real function $f(\mathbf{X}_t,t)$ ($\mathbf{X}_t$ is a stochastic process in $\mathbb{R}^d$), is defined as
\begin{eqnarray}\label{eq:24}
\mathbf{I}_t=\int_{t_0}^{t}f(\mathbf{X}_{t'},t')\dd \mathbf{P}_{t'}^q:=\textrm{\small{${\textrm{ms-lim}}\atop{n\to\infty}$}}\sum_{i=1}^n f(\mathbf{X}_{t_i},t_i) \Delta\mathbf{P}_{t_i}^q,
\end{eqnarray}
where ms-lim denotes the mean square limit
\begin{eqnarray}
\textrm{\small{${\textrm{ms-lim}}\atop{n\to\infty}$}} \mathbf{X}_n=\mathbf{A} \iff \lim_{n\to\infty} \mathbb{E}[|\mathbf{X}_n-\mathbf{A}|^2]=0.\nonumber
\end{eqnarray}
\end{defn}

The following theorem states an useful formula to compute $q$-It\^{o} integrals. In particular, it will be a powerful tool to solve equation (\ref{eq:23}).

\begin{rem}
In what follows, $\mathsf{L}^\mathrm{p}([0,t])$ denotes the set of real functions such that
\begin{eqnarray*}
||f||_{\textrm{$\displaystyle{\mathsf{L^p}([0,t])}$}}=\left(\int_0^t |f(\tau)|^\mathsf{p}\dd \tau\right)^{\frac{1}{\mathsf{p}}}<\infty
\end{eqnarray*}
\end{rem}

\begin{thm}\label{thm:1}
Let $f\in\mathsf{L^p}([0,t])$ be a real function, and let $\mathsf{p}=2+d(1-q)$, $q<\min\{(1+d)/d,(d+2)/(d+4)\}$. The $q$-It\^{o} integral of $f$, 
\begin{eqnarray}\label{eq:25}
\mathbf{I}_t=\int_0^t f(t')\dd {\mathbf{P}}_{t'}^q,
\end{eqnarray}
is a $q$-Gaussian stochastic process with the following properties:
\begin{eqnarray}\label{eq:26}
\mathbb{E}_q(\mathbf{I}_t)=\mathbf{0}, \hspace{0.5cm} \textrm{and} \hspace{0.5cm} \var(\mathbf{I}_t)=\varepsilon_{qd}||f||_{\textrm{$\displaystyle{\mathsf{L^p}([0,t])}$}}^2.
\end{eqnarray}
That is, $\mathbf{I}_t\sim\mathcal{N}_q\left(\mathbf{0},\varepsilon_{qd}||f||_{\textrm{$\displaystyle{\mathsf{L^p}([0,t])}$}}^{2}\right)$.\\
\end{thm}
\begin{proof}
The proof of this result is present in appendix \ref{Appendix:A}.
\end{proof}

In what follows we will apply the developed formalism to some selected problems.

\subsection*{Free Particle: $\mathbf{F}=\mathbf{0}$}

First, let us consider the simplest possible case, in which there is no drift force, \textit{i.e.} $\mathbf{F}=\mathbf{0}$. For this considered case, the associated nonlinear FPE (\ref{eq:1}) is a nonlinear power-law version of the heat equation, that is
\begin{eqnarray}\label{eq:27}
\frac{\partial p}{\partial t}=D\Delta p^\nu.
\end{eqnarray}
The microscopic dynamics, in turn, is governed by the following stochastic differential equation:
\begin{eqnarray}\label{eq:28}
\dd \mathbf{X}_t=(2D)^{\frac{1}{2+d(1-q)}}\dd \mathbf{P}_t^q.
\end{eqnarray}

The solution of Eq. (\ref{eq:27}) may be directly found by the Pattle-Barenblatt ansatz \cite{Frank05,handbooknonlinear}, which is, however, a laborious work. On the other hand, to solve the stochastic differential equation (\ref{eq:28}) is quite simple. In fact, integrating both sides of (\ref{eq:28}) and applying the theorem \ref{thm:1}, it follows that
\begin{eqnarray}\label{eq:29}
\mathbf{X}_t\sim \mathcal{N}_q\left(\mathbf{0},\varepsilon_{qd}(2Dt)^{\frac{2}{2+d(1-q)}}\right),
\end{eqnarray}
if one consider $\mathbf{X}_0=\mathbf{0}$ almost surely. The probability distribution of $\mathbf{X}_t$,
\begin{eqnarray}\label{eq:30}
p(\mathbf{x},t)=\left(\frac{d}{\mathsf{2}_{qd}\pi_{qd}\sigma^2(t)}\right)^\frac{d}{2}\exp_q\left(-\frac{d\mathbf{x}^2}{\mathsf{2}_{qd}\sigma^2(t)}\right),
\end{eqnarray}
where
\begin{eqnarray}\label{eq:31}
\sigma^2(t):=\varepsilon_{qd}(2Dt)^{\frac{2}{2+d(1-q)}},
\end{eqnarray}
is a solution of Eq. (\ref{eq:27}) with initial condition $p(\mathbf{x},0)=\delta(\mathbf{x})$, where $\delta$ denotes the Dirac's delta function.

The power-law dependence of $\sigma^2$ on $t$, $\sigma^2(t)\propto t^\alpha$, is the main indicator of an anomalous diffusive behaviour. Indeed, normal diffusion processes are characterized by a linear dependence of the variance on $t$, that is $\alpha=1$. This induces a natural classification of anomalous diffusion according to $\alpha$: \textit{superdiffusion} corresponding to $\alpha>1$, and \textit{subdiffusion} corresponding to $\alpha<1$. In other words, $\alpha$ may be understood as a "diffusibility quantifier". From (\ref{eq:31}), we conclude that a superdiffusive behaviour will be observed whenever $q>1$, while a subdiffusion will be observed for $q<1$.

\subsection*{Nonextensive Ornstein-Uhlenbeck Process: $\mathbf{F}=-\gamma \mathbf{x}$}

Let us consider now that the drift force has a linear dependence on $\mathbf{x}$, \textit{i.e.} $\mathbf{F}=-\gamma \mathbf{x}$, where $\gamma>0$. In this case, the associated nonlinear FPE is a nonlinear power-law version of the Smoluchowski equation, that is
\begin{eqnarray}\label{eq:32}
\frac{\partial p}{\partial t}=\gamma\mathrm{Div}(\mathbf{x}p)+D\Delta p^\nu.
\end{eqnarray}
The corresponding microscopic dynamics is governed by the following stochastic differential equation:
\begin{eqnarray}\label{eq:33}
\dd \mathbf{X}_t=-\gamma \mathbf{X}_t\dd t+(2D)^{\frac{2}{2+d(1-q)}}\dd \mathbf{P}_t^q.
\end{eqnarray}
As in the free-particle case, it is trivial to find the solution of the above equation. Indeed, with straightforward manipulations and considering $\mathbf{X}_0=\mathbf{0}$ (almost surely), it follows that
\begin{eqnarray}\label{eq:34}
\mathbf{X}_t=\int_0^t (2D)^{\frac{2}{2+d(1-q)}}e^{-\gamma(t-t')}\dd \mathbf{P}_t^q.
\end{eqnarray}
Applying theorem \ref{thm:1} to the expression above we conclude that
\begin{eqnarray}\label{eq:35}
\mathbf{X}_t\sim \mathcal{N}_q(\mathbf{0},\sigma^2(t)),
\end{eqnarray}
where
\begin{align}\label{eq:36}
\sigma^2(t)&=\varepsilon_{qd}\left\{2D\int_{0}^{t}\exp\left[-(2+d(1-q))\gamma (t-t')\right]\dd t'\right\}^{\frac{2}{2+d(1-q)}} \nonumber\\ &=\varepsilon_{qd}\left\{2D\frac{1-e^{-[2+d(1-q)]\gamma t}}{\gamma[2+d(1-q)]}\right\}^{\frac{2}{2+d(1-q)}}.
\end{align}

The process $\mathbf{X}_t$ may be understood as a nonextensive generalization of the Ornstein-Uhlenbeck Process \cite{Gardiner,OU30}. A remarkable property of this system is that it has a well-defined steady state, which corresponds to the situation in which the stochastic forces are balanced with the linear drift force. The variance in the steady state is
\begin{eqnarray}\label{eq:37}
\sigma^2_{\textrm{eq}}:=\lim_{t\to +\infty}\sigma^2(t)=\varepsilon_{qd}\left\{\frac{2D}{\gamma[2+d(1-q)]}\right\}^{\frac{2}{2+d(1-q)}}.
\end{eqnarray}
For $t$ small, on the other hand, the system behaves like a free-particle. In fact, considering $t$ small in (\ref{eq:36}) we have
\begin{eqnarray}\label{eq:38}
\sigma^2(t)\approx \varepsilon_{qd}(2Dt)^{\frac{2}{2+d(1-q)}}.
\end{eqnarray}

\begin{figure}
\centering
\includegraphics{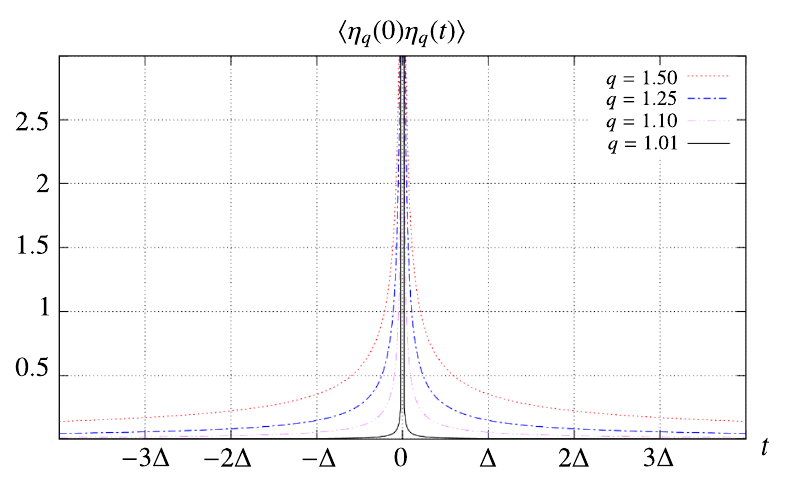}
\caption{Temporal correlations with $q>1$ ($\Delta=0.5$) [Eq.(\ref{eq:45})].}
\label{fig:1}
\end{figure}

\subsection*{Nonextensive Colored Noise}

The effective microscopic dynamics of normal diffusion processes are characterized by stochastic forces that do not exhibit temporal correlations (or memory). So, the influence of such a force at a time $t$ does not depend (or influence) the action at another time $t'\neq t$. This behavior is well modeled by the Gaussian white noise $\eta$, which is heuristically related to the Wiener process as follows
\begin{eqnarray}
\eta(t)\dd t=\dd W_t.
\end{eqnarray}
A well-known fact about $\eta$ (which may be verified by the expression above) is that it has delta-correlation, \textit{i.e.}
\begin{eqnarray}
\langle \eta(t)\eta(t')\rangle =\delta(t-t'),
\end{eqnarray}
which models its (temporal) uncorrelated action.

Herein we are interested on the nonextensive analogue of $\eta$, that is a noise $\eta_q$ such that
\begin{eqnarray}
\eta_q(t)\dd t=\dd \mathrm{P}_t^q,
\end{eqnarray}
where $\mathrm{P}_t^q$ is the one-dimensional generalization of the Wiener process (see definition \ref{defn:1}). The correlators $\langle \eta_q(t)\eta_q(t')\rangle$ may be computed by using the definition of $\mathrm{P}_t^q$. From
\begin{eqnarray}
\mathbb{E}[XY]=\frac{\var[X]+\var[Y]-\var[X-Y]}{2}-\mathbb{E}[X]\mathbb{E}[Y],
\end{eqnarray}
the definition of $\mathrm{P}_t^q$,
\begin{eqnarray}
\var[\mathrm{P}_t^q]=\varepsilon_{q1}t^{\frac{2}{3-q}},
\end{eqnarray}
and $\mathbb{E}[\mathrm{P}_t^q]=0$ (which follows from the fact that one-dimensional $q$-Gaussian distributions are even functions whenever the $q$-expectation value is zero), it follows that
\begin{eqnarray}
\mathbb{E}[\textrm{P}_{t_1}^q\textrm{P}_{t_2}^q]=\frac{\varepsilon_{q1}}{2}\left[t_2^{\frac{2}{3-q}}+t_1^{\frac{2}{3-q}}-(t_2-t_1)^{\frac{2}{3-q}}\right].
\end{eqnarray}
Since $\eta_q(t)\dd t=\dd \mathrm{P}_t^q$, then
\begin{eqnarray*}
\int_0^{t_1}\int_0^{t_2}\langle\eta_q(t')\eta_q(t'')\rangle\dd t'\dd t''=\frac{\varepsilon_{q1}}{2}\left[t_2^{\frac{2}{3-q}}+t_1^{\frac{2}{3-q}}-(t_2-t_1)^{\frac{2}{3-q}}\right],
\end{eqnarray*}
which implies that
\begin{eqnarray}\label{eq:45}
\langle\eta_q(t')\eta_q(t'')\rangle=\varepsilon_{q1}\frac{(q-1)}{(3-q)^2}|t'-t''|^{\frac{2q-4}{3-q}}.
\end{eqnarray}
In the extensive limit we have (see Fig. \ref{fig:1})
\begin{eqnarray}
\lim_{q\to 1^\pm } \langle\eta_q(t')\eta_q(t'')\rangle=\pm\delta(t'-t'').
\end{eqnarray}

From (\ref{eq:45}) we conclude that the stochastic force associated with nonlinear FPEs exhibit memory. These memory effects may explain the anomalous diffusive behavior discussed earlier. Indeed, the correlators $\langle \eta(t+\Delta)\eta(t)\rangle$, $\Delta>0$, are positive if, and only if, $q>1$. So, the action of the stochastic force at an instant $t+\Delta$ is positively correlated with the action at a previous instant $t$, tending to amplify it. Hence, we expect that the diffusive behavior in this case to be greater than the uncorrelated case (which corresponds to normal diffusion with $q=1$). In other words, for $q>1$ we expect a superdiffusive behavior, which agrees with (\ref{eq:31}). A similar reasoning applies for subdiffusion ($q<1$).

\subsection*{Thermal Noise in Electric Circuits}

\begin{figure}
    \centering
    \includegraphics[scale=0.6]{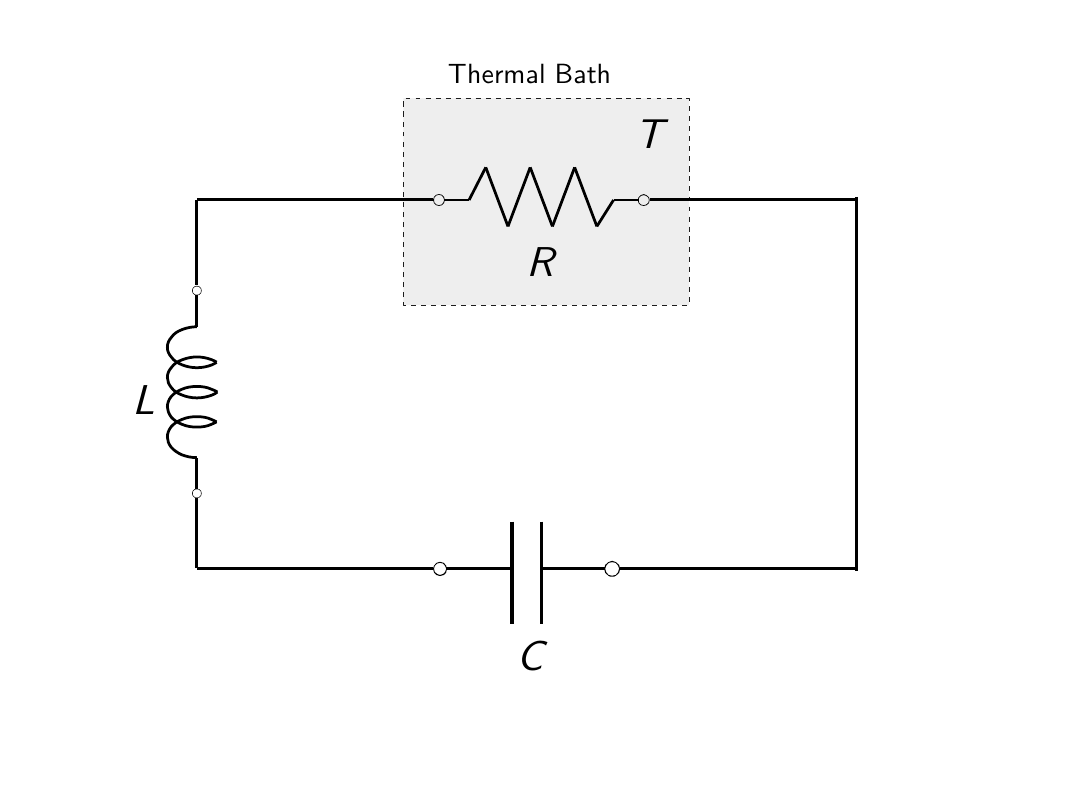}
    \caption{RLC circuit where the resistor $R$ is coupled to a thermal bath at a temperature $T$.}
    \label{fig:2}
\end{figure}

Finally, we will consider the problem of modeling thermal noise in electric circuits. For this, consider an RLC-circuit in series, where the resistor is coupled with a thermal bath at temperature $T$ (see Fig. \ref{fig:2}). In our approach, we will assume a thermal noise modeled by $\eta_q$. The noise-free RLC-circuit is described by the following equations:
\begin{eqnarray}\label{eq:47}
\frac{\dd^2 Q(t)}{\dd t^2}+2\beta\frac{\dd Q(t)}{\dd t}+\omega_0^2 Q(t)=0,
\end{eqnarray}
and
\begin{eqnarray}\label{eq:48}
I(t)=\frac{\dd Q(t)}{\dd t},
\end{eqnarray}
where $Q$ denotes the electric charge, $I$ the electric current, $R$ the resistance, $C$ the capacitance, $L$ the inductance, $\omega_0^2:=1/RL$, and $\beta:=R/2L$.

The dynamics of the circuit with the introduction of the nonextensive thermal noise is governed by two stochastic differential equations:
\begin{eqnarray}\label{eq:49}
\dd I_t=-\omega_0^2 Q_t\dd t-2\beta I_t\dd t+(2D)^{\frac{1}{3-q}}\dd\textrm{P}_{t}^q,
\end{eqnarray}
and
\begin{eqnarray}\label{eq:50}
\dd {Q}_t={I}_t\dd t.
\end{eqnarray}
A strategy to solve the above equations is to define the matrices:
\begin{eqnarray*}
\Lambda=
\begin{bmatrix}
-2\beta & -\omega_0^2 \\ 1 & 0
\end{bmatrix}\hspace{0.25cm}\textrm{and}\hspace{0.25cm} Z_t=\begin{bmatrix}
I_t \\ Q_t
\end{bmatrix},
\end{eqnarray*}
in such a way that (\ref{eq:49}--\ref{eq:50}) reads
\begin{eqnarray}
\dd Z_t=\Lambda{Z}_t\dd t+(2D)^{\frac{1}{3-q}}E_I \dd\textrm{P}_{t}^{q},
\end{eqnarray}
where $E_I=(1,0)^T$ and $E_Q=(0,1)^T$ ($T$ means transposition). With straightforward manipulations one may verify that the solution of the above equation is
\begin{eqnarray}\label{eq:52}
Z_t=e^{\Lambda t}z_0+ (2D)^{\frac{1}{3-q}}\int_0^t e^{(t-t')\Lambda}E_I \dd\textrm{P}_{t}^q,
\end{eqnarray}
where $z_0=(I(0),Q(0))^T$ denotes the initial condition.

\begin{figure}
    \centering
    \includegraphics{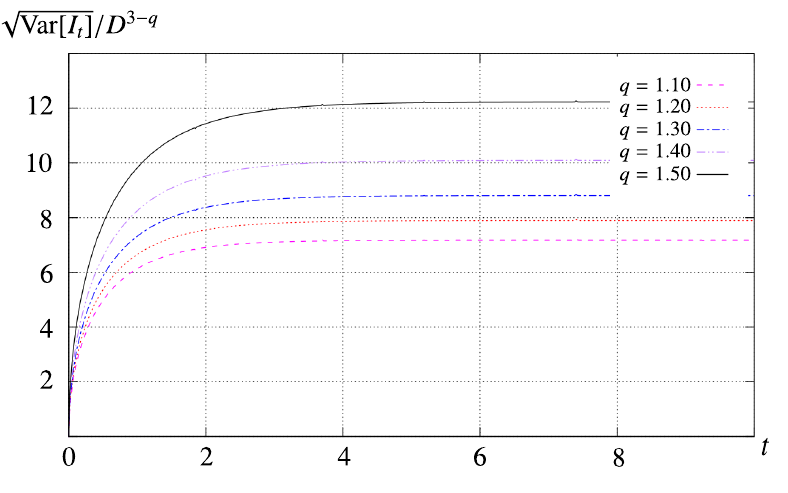}
    \caption{Thermal fluctuations in an RLC-circuit in series with $C=1.0$ mF, $L=50.0$ mH, and $R=1.15$ m$\Omega$.}
    \label{fig:3}
\end{figure}

From (\ref{eq:52}) we can compute all (stochastic) properties of the considered electric circuit. In particular, the current $I_t$ reads
\begin{align}\label{eq:53}
I_t &=E_I^T Z_t \nonumber \\&=E_I^T e^{t\Lambda}z_0+(2D)^{\frac{1}{3-q}}\int_0^t E_I^T e^{(t-t')\Lambda}E_I\dd \textrm{P}_t^q.
\end{align}
We can compute the integral on the right-hand side by diagonalizing the matrix $\exp[(t-t')\Lambda]$ and then applying theorem \ref{thm:1}. The eigenvalues of $\exp[(t-t')\Lambda]$ are $e^{(t-t')(\beta\pm \ii \omega)}$, where $\omega=\sqrt{\beta^2-\omega_0^2}$. From theorem \ref{thm:1}, the variance of $I_t$ is
\begin{widetext}
\begin{eqnarray}
\var[I_t]=\varepsilon_{q1}\left\{\frac{2D}{\omega}\int_0^t \left[\omega\cos\left(\omega(t'-t)\right)+\beta\sin\left(\omega(t'-t)\right)\right]^{3-q}e^{(3-q)\beta(t'-t)}\dd t'\right\}^{\frac{2}{3-q}}.
\end{eqnarray}
\end{widetext}
To understand the behavior of thermal fluctuation in the circuit, we can consider some particular values for the circuit parameters. For instance consider
\begin{eqnarray*}
C=1.0 \textrm{ mF,}\hspace{1cm} L=50.0\textrm{ mH,}\hspace{0.5cm} \textrm{and} \hspace{0.5cm} R=1.15 \textrm{ m$\Omega$}.
\end{eqnarray*}
For these values, the circuit has a sub-critical damping. In Fig. \ref{fig:3} we have the behavior of the variance as a function of time (in minutes). The fluctuations become stationary for times longer than approximately $5$ min, which is due to the fact that the current is approximately zero. The nonextensivity, as we see in Fig. \ref{fig:3}, increases the amplitude of the thermal oscillations for $q>1$. As in superdiffusion, such a phenomenon may be understood as an effect of memory in the nonextensive noise $\eta_q$: the action of such a noise at a time $t+\Delta$, $\Delta>0$, is positively correlated with the action at a time $t$. Hence, the amplitude of the thermal fluctuations will become larger as $q$ increases.

\newpage
\section{Discussion}\label{sec:IV}

In the present paper we proposed a formalism to describe the effective microscopic dynamics of power-law nonlinear Fokker-Planck equations. The formalism is based on the nonextensive generalization of the Wiener process (see Def. \ref{defn:1}). We have demonstrated that, with an adequate generalization of It\^{o}'s integral (see Def. \ref{defn:2}), the associated equation of motion can be trivially solved for important cases, namely with $\mathbf{F}=\mathbf{0}$ (Brownian free-particle) and $\mathbf{F}=-\gamma \mathbf{x}$ (Ornstein-Uhlenbeck process). The proposed formalism also provides an explanation for anomalous diffusive behaviours. In particular, super(sub)-diffusion may be understood in terms of memory effects in the effective stochastic force. We also showed that such a formalism can be easily applied in modeling thermal noise in electric circuits.

From these results we conclude that the proposed formalism is an important tool for a better understanding of nonequilibrium phenomena in nonextensive statistical mechanics. The $q$-Wiener process, in turn, is a new element on the list of $q$-generalizations. 

Although we have shed some light on important questions about nonextensive statistical mechanics and nonlinear Fokker-Planck equations, the results presented in this contribution leave several questions open. Among such issues, we highlight the rigorous construction of the process $\mathbf{P}_t^q$ and the connection between Eq. (\ref{eq:23}) and nonlinear FPE (\ref{eq:1}). We believe that the former can be solved by using the $q$-central limit theorem proposed in Ref. \cite{UTS08}, while the latter requires an $q$-generalization of It\^{o}'s lemma. 

\newpage
\begin{acknowledgments}
The author thanks Gabriel Landi, Constantino Tsallis, Jorge Anderson Ramos, M\'{a}rcio Bortoloti, Luizdarcy Castro, Thiago Mergulh\~{a}o, and Rafael Wagner for the comments and suggestions. This work was supported by Conselho Nacional de Desenvolvimento Cient\'{i}fico e Tecnol\'{o}gico (CNPq).
\end{acknowledgments}

\newpage
\textbf{ }
\newpage

\appendix

\begin{widetext}
\section{Proof of Theorem 1}\label{Appendix:A}

\noindent \textbf{Theorem 1.} \textit{Let $f\in\mathsf{L^p}([0,t])$ be a real function, and let $\mathsf{p}=2+d(1-q)$, $q<\min\{(1+d)/d,(d+2)/(d+4)\}$. The $q$-It\^{o} integral of $f$,}
\begin{eqnarray}
\mathbf{I}_t=\int_0^t f(t')\dd {\mathbf{P}}_{t'}^q,
\end{eqnarray}
\textit{is a $q$-Gaussian stochastic process with the following properties:}
\begin{eqnarray}
\mathbb{E}_q[\mathbf{I}_t]=\mathbf{0}, \hspace{0.5cm} \textit{and} \hspace{0.5cm} \var[\mathbf{I}_t]=\varepsilon_{qd}||f||_{\textrm{$\displaystyle{\mathsf{L^p}([0,t])}$}}^2,
\end{eqnarray}
\textit{that is, $\mathbf{I}_t\sim\mathcal{N}_q\left(0,\varepsilon_{qd}||f||_{\textrm{$\displaystyle{\mathsf{L^p}([0,t])}$}}^{2}\right)$.}

\begin{proof}
Let $N\in\mathbb{Z}^+$, $N>1$, and consider the partial sum $\mathbf{S}_N$,
\begin{eqnarray}\label{eq:A3}
\mathbf{S}_N:=\sum_{k=1}^N f(t_k)\left(\mathbf{P}_{t_{k+1}}^q-\mathbf{P}_{t_k}^q\right)=\sum_{k=1}^N f(t_k)\Delta \mathbf{P}_{t_k}^q,
\end{eqnarray}
where $t_N=t$. Since each term of the sum in the right-hand side has the $q$-expectation value equal to zero, then, from the linearity of $\mathbb{E}_q$, it follows that $\mathbb{E}_q[\mathbf{S}_N]=\mathbf{0}$.

The increments of $\mathbf{P}_t^q$ are $q$-independent random variables (see Def. \ref{defn:1}). So, from (\ref{eq:17}), the $q$-characteristic function of $\mathbf{S}_N$ must satisfy the following condition:
\begin{align}\label{eq:A4}
F_q[\mathbf{S}_N](\mathbf{k})=F_q\left[\sum_{k=1}^N f(t_k)\Delta \mathbf{P}_{t_k}^q\right](\mathbf{k})=F_q[f(t_1)\Delta \mathbf{P}_{t_1}^q]\otimes_{\tilde{q}}\dots \otimes_{\tilde{q}}F_q[f(t_N)\Delta \mathbf{P}_{t_N}^q]
\end{align}
where $\tilde{q}=[2+(d-2)(1-q)]/[2+d(1-q)]$. Each term $F_q[f(t_k)\Delta \mathbf{P}_{t_k}^q](\mathbf{k})$ on the right-hand side of (\ref{eq:A4}) is a $\tilde{q}$-Gaussian distribution (up to normalization), given by (\ref{eq:16}). Hence, $F_q[\mathbf{S}_N](\mathbf{k})$ is $\tilde{q}$-Gaussian of the form:
\begin{align}
F_q[\boldsymbol{S}_N](\mathbf{k})&=\exp_{\tilde{q}}\left\{\frac{\mathsf{2}_{qd}[2+d(1-q)]}{8d(\mathsf{2}_{qd}\pi_{qd}/d)^{d(q-1)}}  \sum_{k=1}^N\left[\var\left(f(t_k)\Delta\mathbf{P}_{t}^q\right)\right]^{1-d(q-1)} |\mathbf{k}|^2 \right\} \nonumber \\
&=\exp_{\tilde{q}}\left\{\frac{\mathsf{2}_{qd}[2+d(1-q)]}{8d(\mathsf{2}_{qd}\pi_{qd}/d)^{d(q-1)}}\sum_{k=1}^N\left[\varepsilon_{qd}|f(t_k)|^2(\Delta t_k)^{\frac{2}{2+d(1-q)}}\right]^{1-d(q-1)} |\mathbf{k}|^2 \right\} \nonumber\\ &=\exp_{\tilde{q}}\left\{\frac{\mathsf{2}_{qd}[2+d(1-q)]}{8d(\mathsf{2}_{qd}\pi_{qd}/d)^{d(q-1)}}\sum_{k=1}^N\left[(\varepsilon_{qd})^{\frac{2+d(1-q)}{2}}|f(t_k)|^{{2+d(1-q)}}\Delta t_k\right]^{\frac{2[1-d(q-1)]}{2+d(1-q)}} |\mathbf{k}|^2 \right\}\nonumber \\ &=\exp_{\tilde{q}}\left\{\frac{\mathsf{2}_{qd}[2+d(1-q)]}{8d(\mathsf{2}_{qd}\pi_{qd}/d)^{d(q-1)}}\sum_{k=1}^N\left[(\varepsilon_{qd})^{\frac{\mathsf{p}}{2}}|f(t_k)|^{\mathsf{p}}\Delta t_k\right]^{\frac{2[1-d(q-1)]}{\mathsf{p}}} |\mathbf{k}|^2 \right\},
\end{align}
with $\mathsf{p}=2+d(1-q)$. Comparing the above equation with (\ref{eq:16}), the variance of $\mathbf{S}_N$ reads
\begin{eqnarray}\label{eq:A6}
\var[\mathbf{S}_N]=\left\{\sum_{k=1}^N\left[(\varepsilon_{qd})^{\frac{\mathsf{p}}{2}}|f(t_k)|^{\mathsf{p}}\Delta t_k\right]^{\frac{2[1-d(q-1)]}{\mathsf{p}}}\right\}^{\frac{1}{1-d(q-1)}}.
\end{eqnarray}
Without loss of generality, consider $\Delta t_k=1/N$. Since $1+d(1-q)>0$ (because $q<(1+d)/d$), then
\begin{align}
\var[\mathbf{I}_t]=\lim_{N\to \infty} \var[\mathbf{S}_N]&=\lim_{N\to \infty} \left\{\sum_{k=1}^N\left[(\varepsilon_{qd})^{\frac{\mathsf{p}}{2}}|f(t_k)|^{\mathsf{p}}\Delta t_k\right]^{\frac{2[1-d(q-1)]}{\mathsf{p}}}\right\}^{\frac{1}{1-d(q-1)}}\nonumber\\ &=\lim_{N\to \infty}\varepsilon_{qd}\left\{\sum_{k=1}^N |f(t_k)|^\mathsf{p}\Delta t_k\right\}^{\frac{2}{\mathsf{p}}} \nonumber \\ &=\varepsilon_{qd}\left\{\lim_{N\to \infty}\sum_{k=1}^N |f(t_k)|^\mathsf{p}\Delta t_k\right\}^{\frac{2}{\mathsf{p}}} \nonumber \\ &=\varepsilon_{qd}\left(\int_0^t |f(\tau)|^{\mathsf{p}}\dd \tau\right)^{\frac{2}{\mathsf{p}}},
\end{align}
which proves the theorem. 
\end{proof}
\newpage
\end{widetext}

\end{document}